\documentclass[12pt,fleqn]{article}
\pdfoutput=1
\usepackage{amsmath}
\usepackage{amsthm}
\usepackage{amsthm}
\usepackage{amsfonts}
\usepackage{amssymb}
\usepackage{amscd}
\usepackage{mathtools}
\usepackage{fullpage}
\usepackage{color}
\usepackage{times}
\usepackage{graphicx}

\newtheorem{prop}{Proposition}
\newtheorem{lemma}{Lemma}
\newtheorem{corollary}{Corollary}
\begin{document}
\renewcommand{\thefootnote}{\fnsymbol{footnote}}
\title{On commutativity of Backus and Gazis averages}
\author{
David R. Dalton\footnote{
Department of Earth Sciences, Memorial University of Newfoundland,
{\tt dalton.nfld@gmail.com}},
Michael A. Slawinski \footnote{
Department of Earth Sciences, Memorial University of Newfoundland,
{\tt mslawins@mac.com}}
}
\date{January 12, 2016}
\maketitle
\renewcommand{\thefootnote}{\arabic{footnote}}
\setcounter{footnote}{0}
\section*{Abstract}
We show that the Backus (1962) equivalent-medium average, which is an average over a spatial variable, and the Gazis et al. (1963) effective-medium average, which is an average over a symmetry group, do not commute, in general.
They commute in special cases, which we exemplify.
\section{Introduction}
Hookean solids are defined by their mechanical property relating linearly the stress tensor,~$\sigma$\,, and the strain tensor,~$\varepsilon$\,,
\begin{equation*}
\sigma_{ij}=\sum_{k=1}^3\sum_{\ell=1}^3c_{ijk\ell}\varepsilon_{k\ell}\,,\qquad i,j=1,2,3
\,.
\end{equation*}
The elasticity tensor,~$c$\,, belongs to one of eight material-symmetry classes shown in Figure~\ref{fig:orderrelation}.
\begin{figure}
\begin{center}
\includegraphics[scale=0.7]{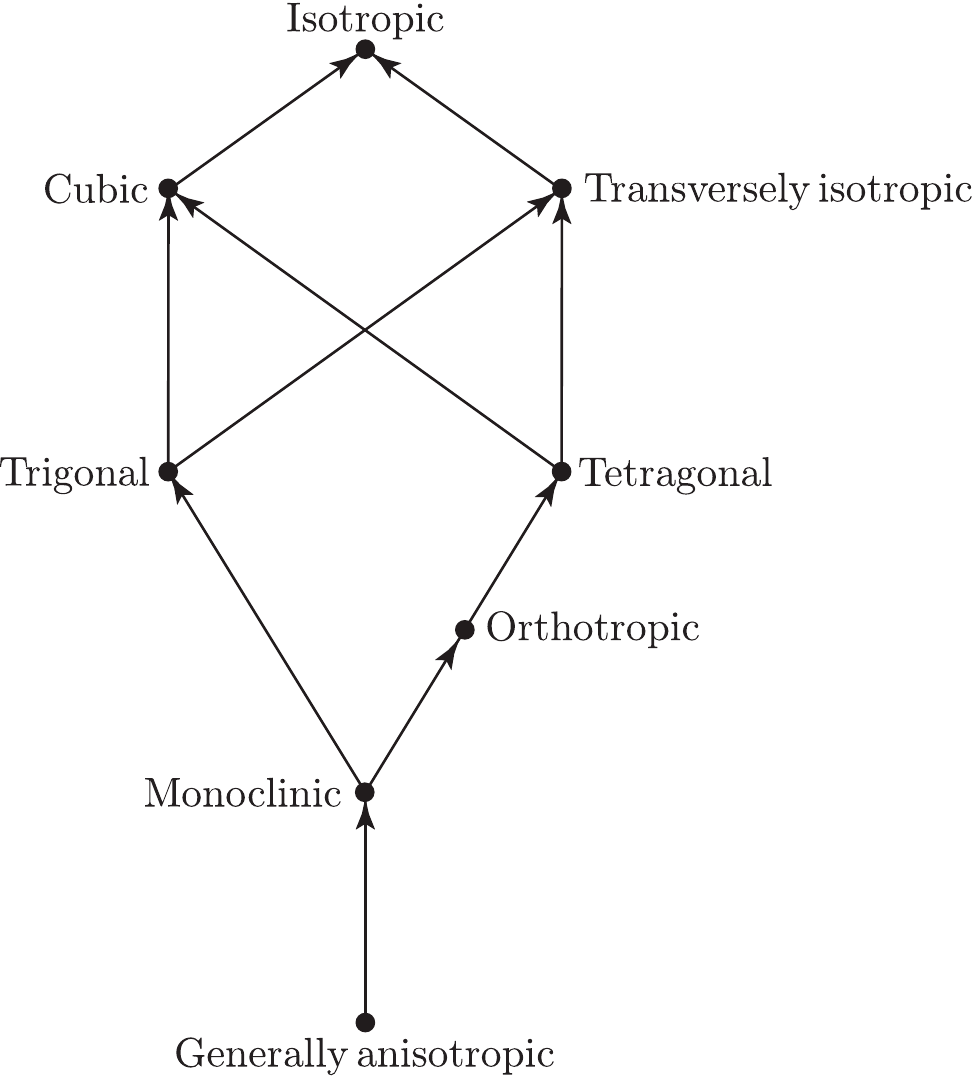}
\end{center}
\caption{\small{Order relation of material-symmetry classes of elasticity tensors:
Arrows indicate subgroups in this partial ordering. For instance, monoclinic is a subgroup of all nontrivial symmetries, in particular, of both orthotropic and trigonal, but orthotropic is not a subgroup of trigonal or {\it vice-versa}.}}
\label{fig:orderrelation}
\end{figure}

The Backus (1962) moving average allows us to quantify the response of a wave propagating through a series of parallel layers whose thicknesses are much smaller than the wavelength.
Each layer is a Hookean solid exhibiting a given material symmetry with given elasticity parameters.
The average is a Hookean solid whose elasticity parameters---and, hence, its material symmetry---allow us to model a long-wavelength response.
This material symmetry of the resulting medium, to which we refer as {\sl equivalent}, is a consequence of symmetries exhibited by the averaged layers.

The long-wave-equivalent medium to a stack of isotropic or transversely isotropic layers
with thicknesses much less than the signal wavelength was
shown by Backus (1962) to be a homogeneous or nearly homogeneous transversely
isotropic medium, where a {\it nearly\/} homogeneous medium is a consequence of a {\it moving\/} average.
Backus (1962) formulation is reviewed by Slawinski (2016) and Bos et al.\ (2016),
where formulations for generally anisotropic, monoclinic, and orthotropic thin layers
are also derived.   Bos et al.\ (2016) examine the underlying assumptions and
approximations behind the Backus (1962) formulation, which is derived by
expressing rapidly varying stresses and strains in terms of products of
algebraic combinations of rapidly varying elasticity parameters with
slowly varying stresses and strains.
The only mathematical approximation in the formulation
is that the average of a product of a rapidly varying function and a slowly
varying function is approximately equal to the product of the averages of
the two functions.

According to Backus (1962), the average of $f(x_3)$ of ``width''~$\ell'$  is
\begin{equation}
\label{eq:BackusOne}
\overline f(x_3):=\int\limits_{-\infty}^\infty w(\zeta-x_3)f(\zeta)\,{\rm d}\zeta
\,,
\end{equation}
where $w(x_3)$ is the weight function with the following properties:
\begin{equation*}
w(x_3)\geqslant0\,,
\quad w(\pm\infty)=0\,,
\quad
\int\limits_{-\infty}^\infty w(x_3)\,{\rm d}x_3=1\,,
\quad
\int\limits_{-\infty}^\infty x_3w(x_3)\,{\rm d}x_3=0\,,
\quad
\int\limits_{-\infty}^\infty x_3^2w(x_3)\,{\rm d}x_3=(\ell')^2\,.
\end{equation*}
These properties define $w(x_3)$ as a probability-density function with mean~$0$ and standard deviation~$\ell'$\,, explaining the use of the term ``width'' for $\ell'$\,.

Gazis et al.~(1963) average allows us to obtain the closest symmetric counterpart---in the Frobenius sense---of a chosen material symmetry to a generally anisotropic Hookean solid.
The average is a Hookean solid, to which we refer as {\sl effective}, whose elasticity parameters correspond to the symmetry chosen {\it a priori}.

Gazis average is a projection given by
\begin{equation}
\widetilde c^{\,\,\rm sym}:=\intop_{G^{\rm sym}}(g\circ c)\,\mathrm{d}\mu(g)
\,,
\label{eq:proj}
\end{equation}
where the integration is over the symmetry group, $G^{\rm sym}$\,, whose elements are $g$\,, with respect to the invariant measure, $\mu$\,, normalized so that $\mu(G^{\rm sym})=1$\,; $\widetilde c^{\,\,\rm sym}$ is the orthogonal projection of $c$\,, in the sense of the Frobenius norm,  on the linear space containing all tensors of that symmetry, which are $ c^{\,\,\rm sym}$\,.
Integral~(\ref{eq:proj}) reduces to a finite sum for the classes whose symmetry groups are finite, which are all classes except isotropy and transverse isotropy. 

The Gazis et al.\  (1963) approach is reviewed and extended by Danek et al.\ (2013, 2015)
in the context of random errors.
Therein, elasticity tensors are not constrained to the same---or even different but known---orientation of the coordinate system.   

Concluding this introduction, let us emphasize that the fundamental distinction between the two averages is their domain of operation.
The Gazis et al.\ (1963) average is an average over symmetry groups at a point and the Backus (1962) average is a spatial average over a distance.
Both averages can be used, separately or together, in quantitative seismology.
Hence, an examination of their commutativity might provide us with an insight into their physical meaning and into allowable mathematical operations.
\section{Generally anisotropic layers and monoclinic medium}
Let us consider a stack of generally anisotropic layers to obtain a monoclinic medium.
To examine the commutativity between the Backus and Gazis averages, let us study the following diagram,
\begin{equation}
\label{eq:CD2}
\begin{CD}
\rm{aniso}@>\rm{B}>>\rm{aniso}\\
@V\mathrm{G}VV                         @VV\rm{G}V\\
\rm{mono}@>>\rm{B}>\rm{mono}
\end{CD}
\end{equation}
and Proposition~\ref{thm:One}, below, 
\begin{prop}
\label{thm:One}
In general, the Backus and Gazis averages do not commute.
\end{prop}
\begin{proof}
To prove this proposition and in view of Diagram~\ref{eq:CD2}, let us begin with the following corollary.
\begin{corollary}
For the generally anisotropic and monoclinic symmetries, the Backus and Gazis averages do not commute.
\end{corollary}
\noindent To understand this corollary, we invoke the following lemma, whose proof is in \ref{AppOne1}.
\begin{lemma}
\label{lem:Mono}
For the effective monoclinic symmetry, the result of the Gazis average is tantamount to replacing each $c_{ijk\ell}$\,, in a generally anisotropic tensor, by its corresponding $c_{ijk\ell}$ of the monoclinic tensor, expressed in the natural coordinate system, including replacements of the anisotropic-tensor components by the zeros of the corresponding monoclinic components. 
\end{lemma}
\noindent Let us first examine the counterclockwise path of Diagram~\ref{eq:CD2}.
Lemma~\ref{lem:Mono} entails a corollary.
\begin{corollary}
\label{col:Mono}
For the effective monoclinic symmetry, given a generally anisotropic tensor,~$C$\,, 
\begin{equation}
\label{eq:GazisMono}
\widetilde{C}^{\,\rm mono}=C^{\,\rm mono}
\,;
\end{equation} 
where $\widetilde{C}^{\,\rm mono}$ is the Gazis average of~$C$\,, and $C^{\,\rm mono}$ is a monoclinic tensor whose nonzero entries are the same as for~$C$\,.
\end{corollary}
\noindent According to Corollary~\ref{col:Mono}, the effective monoclinic
tensor is obtained simply by setting  to zero---in the generally anisotropic tensor---the
components that are zero for the monoclinic tensor. 
Then, the second counterclockwise branch of Diagram~\ref{eq:CD2} is performed as follows.
Applying the Backus average, we obtain (Bos et al., 2015)
\begin{equation*}
\langle c_{3333}\rangle=\overline{\left(\frac{1}{c_{3333}}\right)}^{\,\,-1}\,,
\qquad
\langle c_{2323}\rangle=\frac{\overline{\left(\frac{c_{2323}}{D}\right)}}{2D_2}\,,
\end{equation*} 
\begin{equation*}
\langle c_{1313}\rangle=\frac{\overline{\left(\frac{c_{1313}}{D}\right)}}{2D_2}\,,
\qquad
\langle c_{2313}\rangle=\frac{\overline{\left(\frac{c_{2313}}{D}\right)}}{2D_2}\,,
\end{equation*}
where $D\equiv 2(c_{2323}c_{1313}-c_{2313}^2)$ and
$D_2\equiv (\overline{c_{1313}/D})(\overline{c_{2323}/D})-(\overline{c_{2313}/D})^2$\,.
We also obtain
\begin{equation*}
\langle c_{1133}\rangle=
\overline{\left(\frac{1}{c_{3333}}\right)}^{\,\,-1}
\overline{\left(\frac{c_{1133}}{c_{3333}}\right)}\,,
\quad
\langle c_{2233}\rangle=
\overline{\left(\frac{1}{c_{3333}}\right)}^{\,\,-1}
\overline{\left(\frac{c_{2233}}{c_{3333}}\right)}\,,
\quad
\langle c_{3312}\rangle=
\overline{\left(\frac{1}{c_{3333}}\right)}^{\,\,-1}
\overline{\left(\frac{c_{3312}}{c_{3333}}\right)}\,,
\end{equation*}
\begin{equation*}
\langle c_{1111}\rangle=
\overline{c_{1111}}-\overline{\left(\frac{c_{1133}^2}{c_{3333}}\right)}+
\overline{\left(\frac{1}{c_{3333}}\right)}^{\,\,-1}
\overline{\left(\frac{c_{1133}}{c_{3333}}\right)}^{\,2}\,,
\end{equation*}
\begin{equation*}
\langle c_{1122}\rangle=
\overline{c_{1122}}-\overline{\left(\frac{c_{1133}\,c_{2233}}{c_{3333}}\right)}+
\overline{\left(\frac{1}{c_{3333}}\right)}^{\,\,-1}
\overline{\left(\frac{c_{1133}}{c_{3333}}\right)}\,\,
\overline{\left(\frac{c_{2233}}{c_{3333}}\right)}\,,
\end{equation*}
\begin{equation*}
\langle c_{2222}\rangle=
\overline{c_{2222}}-\overline{\left(\frac{c_{2233}^2}{c_{3333}}\right)}+
\overline{\left(\frac{1}{c_{3333}}\right)}^{\,\,-1}
\overline{\left(\frac{c_{2233}}{c_{3333}}\right)}^{\,2}\,,
\end{equation*}
\begin{equation*}
\langle c_{1212}\rangle=
\overline{c_{1212}}-\overline{\left(\frac{c_{3312}^2}{c_{3333}}\right)}+
\overline{\left(\frac{1}{c_{3333}}\right)}^{\,\,-1}
\overline{\left(\frac{c_{3312}}{c_{3333}}\right)}^{\,2}\,,
\end{equation*}
\begin{equation*}
\langle c_{1112}\rangle=
\overline{c_{1112}}-\overline{\left(\frac{c_{3312}\,c_{1133}}{c_{3333}}\right)}+
\overline{\left(\frac{1}{c_{3333}}\right)}^{\,\,-1}
\overline{\left(\frac{c_{1133}}{c_{3333}}\right)}\,\,
\overline{\left(\frac{c_{3312}}{c_{3333}}\right)}
\end{equation*}
and
\begin{equation*}
\langle c_{2212}\rangle=
\overline{c_{2212}}-\overline{\left(\frac{c_{3312}\,c_{2233}}{c_{3333}}\right)}+
\overline{\left(\frac{1}{c_{3333}}\right)}^{\,\,-1}
\overline{\left(\frac{c_{2233}}{c_{3333}}\right)}\,\,
\overline{\left(\frac{c_{3312}}{c_{3333}}\right)}\,,
\end{equation*}
where angle brackets denote the equivalent-medium elasticity parameters.
The other equivalent-medium elasticity parameters are zero.

Following the clockwise path of Diagram~\ref{eq:CD2},
the upper branch is derived in matrix form in Bos et al.\ (2015).
Then, from Bos et al. (2015) the result of the right-hand branch is
derived by setting entries in the generally anisotropic tensor 
that are zero for the monoclinic tensor to zero.
The nonzero entries, which are too complicated to display explicitly, are---in general---not the same as the result of the counterclockwise path.
Hence, for generally anisotropic and monoclinic symmetries, the Backus and Gazis averages do not commute.
\end{proof}
\section{Higher symmetries}
\subsection{Monoclinic layers and orthotropic medium}
\label{sec:mono}
Proposition~\ref{thm:One} remains valid for layers exhibiting higher material symmetries, and simpler expressions of the corresponding elasticity tensors allow us to examine special cases that result in commutativity.
Let us consider the following corollary of Proposition~\ref{thm:One}.
\begin{corollary}
\label{thm:Two}
For the monoclinic and orthotropic symmetries, the Backus and Gazis averages do not commute.
\end{corollary}
\noindent To study this corollary, let us consider the following diagram,
\begin{equation}
\label{eq:CD}
\begin{CD}
\rm{mono}@>\rm{B}>>\rm{mono}\\
@V\mathrm{G}VV                         @VV\rm{G}V\\
\rm{ortho}@>>\rm{B}>\rm{ortho}
\end{CD}
\end{equation}
and the lemma, whose proof is in \ref{AppOne2}.
\begin{lemma}
\label{lem:Ortho}
For the effective orthotropic symmetry, the result of the Gazis average is tantamount to replacing each $c_{ijk\ell}$\,, in a generally anisotropic---or monoclinic---tensor, by its corresponding $c_{ijk\ell}$ of the orthotropic tensor, expressed in the natural coordinate system, including the replacements by the corresponding zeros. 
\end{lemma}
\noindent Lemma~\ref{lem:Ortho} entails a corollary.
\begin{corollary}
\label{col:Ortho}
For the effective orthotropic symmetry, given a generally anisotropic---or monoclinic---tensor,~$C$\,, 
\begin{equation}
\label{eq:GazisOrtho}
\widetilde{C}^{\,\rm ortho}=C^{\,\rm ortho}
\,.
\end{equation}
where $\widetilde{C}^{\,\rm ortho}$ is the Gazis average of~$C$\,, and $C^{\,\rm ortho}$ is an orthotropic tensor whose nonzero entries are the same as for~$C$\,.
\end{corollary}
\noindent Let us consider a monoclinic tensor and proceed counterclockwise along the first branch of Diagram~\ref{eq:CD}.
Using the fact that the monoclinic symmetry is a special case of general anisotropy, we invoke Corollary~\ref{col:Ortho} to conclude that $\widetilde{C}^{\,\rm ortho}=C^{\,\rm ortho}$\,,
which is equivalent to setting $c_{1112}$\,, $c_{2212}$\,, $c_{3312}$ and $c_{2313}$ to zero
in the monoclinic tensor.
We perform the upper branch of Diagram~\ref{eq:CD}, which is the averaging of a stack of monoclinic layers to get a monoclinic equivalent medium, as in the case of the lower branch of Diagram~\ref{eq:CD2}.
Thus, following the clockwise path, we obtain
\begin{equation*}
c_{1212}^\circlearrowright=
\overline{c_{1212}}-\overline{\left(\frac{c_{3312}^2}{c_{3333}}\right)}+
\overline{\left(\frac{1}{c_{3333}}\right)}^{\,\,-1}
\overline{\left(\frac{c_{3312}}{c_{3333}}\right)}^{\,2}\,,
\end{equation*}
\begin{equation*}
c_{1313}^\circlearrowright=\overline{\left(\frac{c_{1313}}{D}\right)}/(2D_2)\,,\qquad
c_{2323}^\circlearrowright=\overline{\left(\frac{c_{2323}}{D}\right)}/(2D_2)
\end{equation*}
Following the counterclockwise path, we obtain
\begin{equation*}
c_{1212}^\circlearrowleft=\overline{c_{1212}}\,,\quad
c_{1313}^\circlearrowleft=\overline{\left(\frac{1}{c_{1313}}\right)}^{\,\,-1}\,,\quad
c_{2323}^\circlearrowleft=\overline{\left(\frac{1}{c_{2323}}\right)}^{\,\,-1}\,.
\end{equation*}
The other entries are the same for both paths.

In conclusion, the results of the clockwise and  counterclockwise paths are the same if $c_{2313}=c_{3312}=0$\,, which is a special case of monoclinic symmetry.
Thus, the Backus average and Gazis average commute for that case, but not in general.
\subsection{Orthotropic layers and tetragonal medium}
\label{sec:ortho}
In a manner analogous to Diagram~\ref{eq:CD}, but proceeding from the the upper-left-hand corner orthotropic tensor to lower-right-hand corner tetragonal tensor by the counterclockwise path,
\begin{equation}
\label{eq:CD3}
\begin{CD}
\rm{ortho}@>\rm{B}>>\rm{ortho}\\
@V\mathrm{G}VV                         @VV\rm{G}V\\
\rm{tetra}@>>\rm{B}>\rm{tetra}
\end{CD}
\end{equation}
we obtain
\begin{equation*}
c_{1111}^\circlearrowleft=\overline{\frac{c_{1111}+c_{2222}}{2}-
\frac{\left(\frac{c_{1111}+c_{2222}}{2}\right)^2}{c_{3333}}}+
\overline{\left(\frac{c_{1111}+c_{2222}}{2c_{3333}}\right)}^2
\overline{\left(\frac{1}{c_{3333}}\right)}^{\,\,-1}
\,.
\end{equation*}
Following the clockwise path, we obtain
\begin{equation*}
c_{1111}^\circlearrowright=\overline{\frac{c_{1111}+c_{2222}}{2}-
\frac{c_{1133}^2+c_{2233}^2}{2c_{3333}}}+
\frac{1}{2}\left[\overline{\left(\frac{c_{1133}}{c_{3333}}\right)}^2+
\overline{\left(\frac{c_{2233}}{c_{3333}}\right)}^2\right]
\overline{\left(\frac{1}{c_{3333}}\right)}^{\,\,-1}\,.
\end{equation*}
These results are not equal to one another, unless $c_{1133}=c_{2233}$\,, which is a special case of orthotropic symmetry.
Also $c_{2323}$ must equal $c_{1313}$ for $c_{2323}^\circlearrowright=c_{2323}^\circlearrowleft$.
The other entries are the same for both paths.
Thus, the Backus average and Gazis average do commute for $c_{1133}=c_{2233}$ and $c_{2323}=c_{1313}$\,, which is a special case of orthotropic symmetry, but not in general.

Let us also consider the case of monoclinic layers and a tetragonal medium to examine the process of combining the Gazis averages, which is tantamount to combining Diagrams~(\ref{eq:CD}) and~(\ref{eq:CD3}),
\begin{equation}
\begin{CD}
\label{eq:CD4}
\rm{mono}@>\rm{B}>>\rm{mono}\\
@V\mathrm{G}VV                         @VV\rm{G}V\\
\rm{ortho}@>>\rm{B}>\rm{ortho}\\
@V\mathrm{G}VV                         @VV\rm{G}V\\
\rm{tetra}@>>\rm{B}>\rm{tetra}
\end{CD}
\end{equation}
In accordance with Proposition~\ref{thm:One}, there is---in general---no commutativity.
However, the outcomes are the same as for the corresponding steps in Sections~\ref{sec:mono} and  \ref{sec:ortho}.
In general, for the Gazis average, proceeding directly, $\rm{aniso}\xrightarrow{\rm{G}}\rm{iso}$\,, is tantamount to proceeding along arrows in Figure~\ref{fig:orderrelation}, $\rm{aniso}\xrightarrow{\rm{G}}\cdots\xrightarrow{\rm{G}}\rm{iso}$\,.
No such combining of the Backus averages is possible, since, for each step, layers become a homogeneous medium.
\subsection{Transversely isotropic layers}
Lack of commutativity can also be exemplified by the case of transversely isotropic layers.
Following the clockwise path of Diagram~\ref{eq:CD}, the Backus average results in a transversely isotropic medium, whose Gazis average---in accordance with Figure~\ref{fig:orderrelation}---is isotropic.
Following the counterclockwise path, Gazis average results in an isotropic medium, whose Backus average, however, is transverse isotropy.
Thus, not only the elasticity parameters, but even the resulting material-symmetry classes differ.

Also, we could---in a manner analogous to the one illustrated in Diagram~\ref{eq:CD4}\,---begin with generally anisotropic layers and obtain isotropy by the clockwise path and transverse isotropy by the counterclockwise path, which again illustrates noncommutativity.
\section{Discussion}
Herein, we assume that all tensors are expressed in the same orientation of their coordinate systems. Otherwise, the process of averaging become more complicated, as discussed---for the Gazis average---by Kochetov and Slawinski (2009a, 2009b) and as mentioned---for the Backus average---by Bos et al. (2016).

Mathematically, the noncommutativity of two distinct averages is shown by Proposition~\ref{thm:One}, and exemplified for several material symmetries.

We do not see a physical justification for special cases in which---given the same orientation of coordinate systems---these averages commute.
This behaviour might support the view that a mathematical realm, which allows for fruitful analogies with the physical world, has no causal connection with it.
\section*{Acknowledgments}
We wish to acknowledge discussions with Theodore Stanoev. This research was performed in the context of The Geomechanics Project
supported by Husky Energy. Also, this research was partially supported by the
Natural Sciences and Engineering Research Council of Canada, grant 238416-2013.
\section*{References}
\frenchspacing
\newcommand{\hd}{\par\noindent\hangindent=0.4in\hangafter=1}
\hd
Backus, G.E.,  Long-wave elastic anisotropy produced by horizontal layering,
{\it  J. Geophys. Res.\/}, {\bf 67}, 11, 4427--4440, 1962.
\setlength{\parskip}{4pt}
\hd
B\'{o}na, A., I. Bucataru and M.A. Slawinski, Space of $SO(3)$-orbits of elasticity tensors, {\it  Archives of Mechanics\/}, {\bf 60}, 2, 121--136, 2008
\hd 
Bos, L, D.R. Dalton, M.A. Slawinski and T. Stanoev,
On Backus average for generally anisotropic layers, {\it arXiv\/}, 2016.
\hd
Chapman, C. H., {\it Fundamentals of seismic wave propagation\/},  Cambridge 
University Press, 2004.
\hd
Danek, T., M. Kochetov and M.A. Slawinski,  Uncertainty analysis of
effective elasticity tensors using quaternion-based global optimization and 
Monte-Carlo method, {\it The Quarterly Journal of Mechanics and Applied 
Mathematics\/}, {\bf 66}, 2, pp. 253--272, 2013.
\hd
Danek, T., M. Kochetov and M.A. Slawinski, Effective elasticity tensors 
in the context of random errors, {\it Journal of Elasticity\/}, 2015.
\hd
Gazis, D.C., I. Tadjbakhsh and R.A. Toupin, The elastic tensor of given symmetry nearest to an anisotropic elastic tensor, {\it Acta Crystallographica\/}, {\bf 16}, 9, 917--922, 1963.
\hd
Kochetov, M. and M.A. Slawinski, On obtaining effective orthotropic 
elasticity tensors, {\it The Quarterly Journal of Mechanics and Applied 
Mathematics\/}, {\bf 62}, 2, pp. 149-Ð166, 2009a.
\hd
Kochetov, M. and M.A. Slawinski,  On obtaining effective transversely 
isotropic elasticity tensors, {\it Journal of Elasticity\/}, {\bf 94}, 1Ð-13., 2009b.
\hd
Slawinski, M.A. {\it Wavefronts and rays in seismology: Answers to unasked questions\/},
World Scientific, 2016.
\hd
Slawinski, M.A., {\it Waves and rays in elastic continua\/}, World Scientific, 2015.
\hd
Thomson, W., {\it Mathematical and physical papers: Elasticity, heat, 
electromagnetism\/}, Cambridge University Press, 1890
\setcounter{section}{0}
\setlength{\parskip}{0pt}
\renewcommand{\thesection}{Appendix~\Alph{section}}
\section{}
\subsection{}\label{AppOne1}
Let us prove Lemma~\ref{lem:Mono}.
\begin{proof}
For discrete symmetries, we can write integral~(\ref{eq:proj}) as a sum,
\begin{equation}
\label{eq:AverageDisc}
\widetilde C^{\,\rm sym}=\frac{1}{n}\left(\tilde{A}_1^{\rm sym}\,C\,\tilde{A}_1^{\rm sym}\,{}^{^T}+\ldots+\tilde{A}_n^{\rm sym}\,C\,\tilde{A}_n^{\rm sym}\,{}^{^T}\right)
\,,
\end{equation}
where $\widetilde C^{\rm sym}$ is expressed in Kelvin's notation, in view
of  Thomson (1890, p.~110) as discussed in Chapman (2004, Section~4.4.2).

To write the elements of the monoclinic symmetry group as $6\times 6$ matrices, we must consider orthogonal transformations in $\mathbb{R}^3$\,.
Transformation $A\in SO(3)$ of $c_{ijk\ell}$ corresponds to transformation of $C$ given by
\begin{equation}
{\footnotesize
\tilde{A}=\left[\begin{array}{cccccc}
A_{11}^{2} & A_{12}^{2} & A_{13}^{2} & \sqrt{2}A_{12}A_{13} & \sqrt{2}A_{11}A_{13} & \sqrt{2}A_{11}A_{12}\\
A_{21}^{2} & A_{22}^{2} & A_{23}^{2} & \sqrt{2}A_{22}A_{23} & \sqrt{2}A_{21}A_{23} & \sqrt{2}A_{21}A_{22}\\
A_{31}^{2} & A_{32}^{2} & A_{33}^{2} & \sqrt{2}A_{32}A_{33} & \sqrt{2}A_{31}A_{33} & \sqrt{2}A_{31}A_{32}\\
\sqrt{2}A_{21}A_{31} & \sqrt{2}A_{22}A_{32} & \sqrt{2}A_{23}A_{33} & A_{23}A_{32}+A_{22}A_{33} & A_{23}A_{31}+A_{21}A_{33} & A_{22}A_{31}+A_{21}A_{32}\\
\sqrt{2}A_{11}A_{31} & \sqrt{2}A_{12}A_{32} & \sqrt{2}A_{13}A_{33} & A_{13}A_{32}+A_{12}A_{33} & A_{13}A_{31}+A_{11}A_{33} & A_{12}A_{31}+A_{11}A_{32}\\
\sqrt{2}A_{11}A_{21} & \sqrt{2}A_{12}A_{22} & \sqrt{2}A_{13}A_{23} & A_{13}A_{22}+A_{12}A_{23} & A_{13}A_{21}+A_{11}A_{23} & A_{12}A_{21}+A_{11}A_{22}\end{array}\right]}
\,,
\label{eq:ATildeQ}
\end{equation}
which is an orthogonal matrix, $\tilde{A}\in SO(6)$ (Slawinski (2015), Section~5.2.5).\footnote{Readers interested in formulation of matrix~(\ref{eq:ATildeQ}) might refer to B\'ona et al. (2008).}

The required symmetry-group elements are
\begin{equation*}
 A_1^{\rm mono}=
\left[
\begin{array}{ccc}
1 & 0 & 0\\
0 & 1 & 0\\
0 & 0 & 1\end{array}\right]
\mapsto
\left[\begin{array}{cccccc}
1 & 0 & 0 & 0 & 0 & 0\\
0 & 1 & 0 & 0 & 0 & 0\\
0 & 0 & 1 & 0 & 0 & 0\\
0 & 0 & 0 & 1 & 0 & 0\\
0 & 0 & 0 & 0 & 1 & 0\\
0 & 0 & 0 & 0 & 0 & 1\end{array}\right]
=\tilde{A}_1^{\rm mono}
\end{equation*}
\begin{equation*}
 A_2^{\rm mono}=
\left[
\begin{array}{ccc}
-1 & 0 & 0\\
0 & -1 & 0\\
0 & 0 & 1\end{array}\right]
\mapsto
\left[\begin{array}{cccccc}
1 & 0 & 0 & 0 & 0 & 0\\
0 & 1 & 0 & 0 & 0 & 0\\
0 & 0 & 1 & 0 & 0 & 0\\
0 & 0 & 0 & -1 & 0 & 0\\
0 & 0 & 0 & 0 & -1 & 0\\
0 & 0 & 0 & 0 & 0 & 1\end{array}\right]
=\tilde{A}_2^{\rm mono}
\,.
\end{equation*}
For the monoclinic\index{material symmetry!monoclinic} case, expression~(\ref{eq:AverageDisc}) can be stated explicitly as
\begin{equation*}
\widetilde C^{\rm mono}=
\frac{\left(\tilde{A}_1^{\rm mono}\right)\,C\,\left(\tilde{A}_1^{\rm mono}\right)^T+\left(\tilde{A}_2^{\rm mono}\right)\,C\,\left(\tilde{A}_2^{\rm mono}\right)^T}{2}
\,.
\end{equation*}
Performing matrix operations, we obtain
\begin{equation}
\widetilde C^{\rm mono}
=\left[\begin{array}{cccccc}
c_{1111} & c_{1122} & c_{1133} & 0 & 0 & \sqrt{2}c_{1112}\\
c_{1122} & c_{2222} & c_{2233} & 0 & 0 & \sqrt{2}c_{2212}\\
c_{1133} & c_{2233} & c_{3333} & 0 & 0 & \sqrt{2}c_{3312}\\
0 & 0 & 0 & 2c_{2323} & 2c_{2313} & 0\\
0 & 0 & 0 & 2c_{2313} & 2c_{1313} & 0\\
\sqrt{2}c_{1112} & \sqrt{2}c_{2212} & \sqrt{2}c_{3312} & 0 & 0 & 2c_{1212}
\end{array}\right]
\,,
\label{eq:MonoExplicitRef}
\end{equation}
which exhibits the form of the monoclinic tensor in its natural coordinate system.
In other words, $\widetilde{C}^{\rm mono}=C^{\rm mono}$\,, in accordance with Corollary~\ref{col:Mono}.
\subsection{}\label{AppOne2}
Let us prove Lemma~\ref{lem:Ortho}.

For orthotropic symmetry, $\tilde{A}_1^{\rm ortho}=\tilde{A}_1^{\rm mono}$ and
$\tilde{A}_2^{\rm ortho}=\tilde{A}_2^{\rm mono}$ and
\begin{equation*}
 A_3^{\rm ortho}=
\left[
\begin{array}{ccc}
-1 & 0 & 0\\
0 &  1 & 0\\
0 & 0 & -1\end{array}\right]
\mapsto
\left[\begin{array}{cccccc}
1 & 0 & 0 & 0 & 0 & 0\\
0 & 1 & 0 & 0 & 0 & 0\\
0 & 0 & 1 & 0 & 0 & 0\\
0 & 0 & 0 & -1 & 0 & 0\\
0 & 0 & 0 & 0 &  1 & 0\\
0 & 0 & 0 & 0 & 0 & -1\end{array}\right]
=\tilde{A}_3^{\rm ortho}
\,,
\end{equation*}
\begin{equation*}
 A_4^{\rm ortho}=
\left[
\begin{array}{ccc}
1 & 0 & 0\\
0 &  -1 & 0\\
0 & 0 & -1\end{array}\right]
\mapsto
\left[\begin{array}{cccccc}
1 & 0 & 0 & 0 & 0 & 0\\
0 & 1 & 0 & 0 & 0 & 0\\
0 & 0 & 1 & 0 & 0 & 0\\
0 & 0 & 0 & 1 & 0 & 0\\
0 & 0 & 0 & 0 &  -1 & 0\\
0 & 0 & 0 & 0 & 0 & -1\end{array}\right]
=\tilde{A}_4^{\rm ortho}
\,.
\end{equation*}
For the orthotropic\index{material symmetry!orthotropic} case, expression~(\ref{eq:AverageDisc}) can be stated explicitly as
{\footnotesize
\begin{equation*}
\widetilde C^{\rm ortho}=
\frac{\left(\tilde{A}_1^{\rm ortho}\right)\,C\,\left(\tilde{A}_1^{\rm ortho}\right)^T+\left(\tilde{A}_2^{\rm ortho}\right)\,C\,\left(\tilde{A}_2^{\rm ortho}\right)^T
+\left(\tilde{A}_3^{\rm ortho}\right)\,C\,\left(\tilde{A}_3^{\rm ortho}\right)^T+\left(\tilde{A}_4^{\rm ortho}\right)\,C\,\left(\tilde{A}_4^{\rm ortho}\right)^T
}
{4}
\,.
\end{equation*}}
Performing matrix operations, we obtain
\begin{equation}
\widetilde C^{\rm ortho}
=\left[\begin{array}{cccccc}
c_{1111} & c_{1122} & c_{1133} & 0 & 0 & 0\\
c_{1122} & c_{2222} & c_{2233} & 0 & 0 & 0\\
c_{1133} & c_{2233} & c_{3333} & 0 & 0 &0\\
0 & 0 & 0 & 2c_{2323} & 0 & 0\\
0 & 0 & 0 & 0 & 2c_{1313} & 0\\
0& 0 & 0 & 0 & 0 & 2c_{1212}
\end{array}\right]
\,,
\label{eq:OrthoExplicitRef}
\end{equation}
which exhibits the form of the orthotropic tensor in its natural coordinate system.
In other words, $\widetilde{C}^{\rm ortho}=C^{\rm ortho}$\,, in accordance with Corollary~\ref{col:Ortho}.
\end{proof}
\end{document}